\documentclass[letterpaper, 10pt, conference]{ieeeconf}

\usepackage{amsmath,amssymb,paralist,subfigure,cite,graphicx,color,stmaryrd,mathrsfs,url,algorithm2e}
\usepackage[table]{xcolor}
\usepackage{cuted,mathtools,lipsum}
\newtheorem{lem}{Lemma}
\newtheorem{ass}{Assumption}
\newtheorem{thm}{Theorem}

\newtheorem{rem}{Remark}

\def\mb{\mathbf}
\def\mbb{\mathbb}
\def\mc{\mathcal}

\DeclareMathOperator*{\argmin}{argmin}

\IEEEoverridecommandlockouts                             
\begin{document}
\title{\huge \bf Distributed support-vector-machine over \\ dynamic balanced directed networks}
\author{Mohammadreza Doostmohammadian, Alireza Aghasi,  Themistoklis Charalambous, and  Usman~A.~Khan
\thanks{MD is with the Faculty of Mechanical Engineering at Semnan University, Semnan, Iran, and with the School of Electrical Engineering at Aalto University, Espoo, Finland, \texttt{mohammadreza.doostmohammadian@aalto.fi}. AA is with the Institute for Insight,  Robinson College of Business at Georgia State University, GA, USA, \texttt{aaghasi@gsu.edu}. TC is with the School of Electrical Engineering at Aalto University,
Espoo, Finland, \texttt{themistoklis.charalambous@aalto.fi}. UAK is with the Electrical and Computer Engineering Department at Tufts University, MA, USA, \texttt{khan@ece.tufts.edu}. The work of UAK is partially supported by NSF under awards~\#1903972 and~\#1935555.}}
\maketitle

\begin{abstract}
	In this paper, we consider the binary classification problem via distributed Support-Vector-Machines (SVM), where the idea is to train a network of agents, with limited share of data, to cooperatively learn the SVM classifier for the global database. Agents only share processed information regarding the classifier parameters and the gradient of the local loss functions instead of their raw data. In contrast to the existing work, we propose a continuous-time algorithm that incorporates network topology changes in discrete jumps. This hybrid nature allows us to remove chattering that arises because of the discretization of the underlying CT process. We show that the proposed algorithm converges to the SVM classifier over time-varying weight balanced directed graphs by using  arguments from the matrix perturbation theory. 
	
\end{abstract}

\section{Introduction} \label{sec_intro}
Machine learning has been an area of significant research in recent signal processing and control literature~\cite{xin2020decentralized,xin2020general,giannakis2020time,TCNS_resource}. Among the topics of interest, Support Vector Machines (SVMs) are supervised-learning methods with several applications ranging from image/video processing to bioinformatics. Motivated by the recent progress in computing hardware and wireless communication, we are interested in developing distributed solutions for SVM classification. The basic idea is to process the raw data at each node in order to train a local classifier and then fuse these classifiers among the neighboring nodes. D-SVM (distributed SVM) finds applications where a subset of the data is acquired by different nodes/servers/agents possibly at different geographic locations, privacy is of concern, and communicating data to a fusion center (FC) is practically infeasible.

In binary classification, SVM defines the maximum-margin hyperplane (as the classifier) determined by the closest data samples known as the Support-Vectors (SVs). The preliminary work on D-SVM (referred as Distributed Parallel SVM (DP-SVM)~\cite{lu2008distributed} and Parallel SVM (P-SVM)~\cite{chang2011psvm}) is focused on local calculations and sharing of the SVs, as the representatives of discriminant information of
the dataset~\cite{navia2006distributed,lu2008distributed,chang2011psvm,bordes2005fast,forero2010consensus}. ${\mbox{In~\cite{navia2006distributed,chang2011psvm,bordes2005fast}}}$, these local SVs are updated via a FC to improve the D-SVM performance. Ref.~\cite{lu2008distributed} implements D-SDV on multi-agent networks and requires a Hamiltonian cycle that visits every agent exactly once. A FC-free approach is considered in~\cite{forero2010consensus}, where every agent locally solves a (coupled) convex optimization sub-problem via alternating direction method of multipliers (ADMM), which is not computationally efficient. A main drawback is that these approaches require sharing raw data over the communication network, raising data privacy and information security issues. More recently, consensus-based distributed optimization methods are proposed in~\cite{gharesifard2013distributed,ning2017distributed,garg2019fixed2,rahili_ren,taes2020finite,li2020time,armand2017globally,srivastava2018distributed,mansoori2019fast}, where instead of raw data, agents share \textit{processed} information, which in case of leakage to unauthorized parties reveals little information about the original data. Among these, the solution in~\cite{armand2017globally} requires distributed computation of the Hessian inverse, which is not practical since even for a sparse Hessian matrix, its inverse is not necessarily sparse. Penalty-based approaches are proposed in~\cite{srivastava2018distributed,mansoori2019fast}, where the constrained convex cost function is reformulated by adding a new penalty term on consensus constraint violation. It is shown that there is a gap of~$\mc{O}(\kappa)$ (with~$\kappa$ as the penalty constant) between the optimal penalty-based solution and the original constrained one~\cite{yuan2016convergence}.
Other methods include finite/fixed-time algorithms~\cite{ning2017distributed,garg2019fixed2,rahili_ren,taes2020finite,li2020time} that are prone to
steady-state oscillations (known as \textit{chattering}) due to non-Lipschitz dynamics. 
%In terms of data privacy, these solutions~\cite{gharesifard2013distributed,ning2017distributed,garg2019fixed2,rahili_ren,taes2020finite,armand2017globally,srivastava2018distributed,mansoori2019fast} outperform the classical ones~\cite{navia2006distributed,lu2008distributed} since no raw data (including the SVs) are shared over the network. 

In this paper, a D-SVM method is proposed that overcomes the challenges of (semi-centralized) FC-based solution and the chattering phenomena. Moreover, in contrast to Refs.~\cite{gharesifard2013distributed,ning2017distributed,garg2019fixed2,rahili_ren,taes2020finite,li2020time,armand2017globally,srivastava2018distributed,mansoori2019fast}, where either continuous-time (CT) or discrete-time (DT) protocols are considered, we propose a hybrid algorithm to address the topology switching of the multi-agent network in DT incorporated in a CT gradient-descent update~\cite{goebel2009hybrid}. Our hybrid approach enables more flexibility in considering mixed-dynamics\cite{ly2012learning,goebel2009hybrid}, which allows solving D-SVM via CT protocols over general \textit{dynamic} digraphs in DT domain. To analyze the proposed hybrid model, we use \textit{matrix perturbation theory}~\cite{stewart1990matrix} to characterize the eigenspectrum of the proposed dynamics. The proposed D-SVM is fully distributed, as opposed to FC-based approaches, and does not require solving convex sub-problems unlike~\cite{navia2006distributed,lu2008distributed,chang2011psvm,bordes2005fast,forero2010consensus}. Due to Lipschitz-continuity of the proposed CT approach, it's DT approximation is free of the aforementioned \textit{chattering} inherent to the non-Lipschitz dynamics~\cite{ning2017distributed,garg2019fixed2,rahili_ren,taes2020finite,li2020time}. Note that we directly solve the original constrained optimization  free of penalty-based approximation inaccuracies in~\cite{srivastava2018distributed,mansoori2019fast,yuan2016convergence}.

% To address D-SVM, we first solve the \textit{consensus-constrained} distributed optimization over general  \textit{weight-balanced directed graphs} (WB-digraphs). As compared to constrained distributed optimization over \textit{static} WB-digraphs~\cite{gharesifard2013distributed} and undirected networks~\cite{ning2017distributed,garg2019fixed2,rahili_ren,taes2020finite,armand2017globally,srivastava2018distributed,mansoori2019fast}, this work extends the literature on \textit{unconstrained} distributed optimization and network resource allocation over WB-digraphs~\cite{liang2018distributed,qu2017harnessing} to the consensus-constrained case over \textit{dynamic} WB-digraphs. Our proposed continuous-time/discrete-time (CT/DT) dynamics is modeled  as a hybrid dynamical system to address the dynamic nature of both the multi-agent network and the gradient-descent approach. This work, as in many recent works~\cite{gharesifard2013distributed,ning2017distributed,garg2019fixed2,rahili_ren,taes2020finite,armand2017globally,srivastava2018distributed,mansoori2019fast}, requires the cost function to be (at least) strictly convex and continuously twice-differentiable. 
%Our method is fully distributed as every agent only has access to its own local data-set and the shared estimates of the classifier parameters and local gradients sent from its direct incoming neighbors. Moreover, our solution is free of the so-called chattering inherent to DT approximation of non-Lipschitz dynamics~\cite{ning2017distributed,garg2019fixed2,rahili_ren,taes2020finite}.

We now describe the rest of the paper. Section~\ref{sec_pre} recaps some preliminaries on algebraic graph theory while Section~\ref{sec_prob} formulates the D-SVM problem. Section~\ref{sec_dyn} states our CT gradient descent method to address D-SVM whereas the convergence analysis over dynamic WB-digraphs is available in Section~\ref{sec_conv}. Section~\ref{sec_sim} provides an illustrative example, and finally, Section~\ref{sec_conclusion} concludes the paper with some future research directions.   

\section{Preliminaries on Algebraic Graph Theory} \label{sec_pre}
We represent the multi-agent network by a strongly-connected directed graph (SC digraph)~$\mc{G}$.
Assuming a positive weight~$w_{ij}$ for every link (from node~$j$ to node~$i$) and zero otherwise,  the irreducible adjacency matrix of~$\mc G$ is $W=\{w_{ij}\}$, and the Laplacian matrix~$\overline{W}=\{\overline{w}_{ij}\}$ is defined as, 
\begin{eqnarray} \label{eq_laplac}
      \overline{w}_{ij}=\left \{ \begin{array}{ll} w_{ij} & ,~ i\neq j, \\ -\sum_{j=1}^n w_{ij} & ,~ i=j. \end{array}\right.
\end{eqnarray}
The SC property of the graph is directly related to the rank of its Laplacian matrix as given in the next lemma.
\begin{lem} \label{lem_sc}
	\cite{SensNets:Olfati04} The given Laplacian $\overline{W}$ in \eqref{eq_laplac} for a SC digraph
	has  eigenvalues whose real-parts are non-positive with one isolated eigenvalue at zero.
\end{lem}
Next, we define a WB-digraph as an SC digraph with equal weight-sum of incoming and outgoing links at every node $i$, i.e.,~$\sum_{j=1}^n w_{ji} = \sum_{j=1}^n w_{ij}$,
implying the following lemma.  
\begin{lem} \label{lem_laplacian}
	\cite{SensNets:Olfati04} For the Laplacian~$\overline{W}$ of a WB-digraph, the vectors~$\mb{1}_n^\top$ and~$\mb{1}_n$ are respectively the left and right eigenvector associated with the zero eigenvalue, i.e.,~$\mb{1}_n^\top \overline{W}= \mb{0}_n$ and~$\overline{W}\mb{1}_n=\mb{0}_n$, where~$\mb{1}_n$ and~$\mb{0}_n$ are the column vectors of~$1$'s and~$0$'s of size~$n$, respectively.
\end{lem}

In the rest of the paper,~${\lVert A\rVert_{\infty}}$ denotes the infinity norm of a matrix, i.e.,~$\lVert A\rVert_{\infty}= \max_{1\leq i\leq n} \sum_{j=1}^n |a_{ij}|.$

\section{Problem Statement} \label{sec_prob}
Consider a binary classification problem for a given set of~$N$ data points~$\boldsymbol{\chi}_i \in \mathbb{R}^{m-1}$,~$i=1,\ldots,N$, each belonging to one of two classes labeled by~${l_i \in \{-1,1\}}$. Using the entire training set, the SVM problem is to find a hyperplane~${\boldsymbol{\omega}^\top \boldsymbol{\chi} - \nu =0}$,  for~${\boldsymbol{\chi}\in\mbb R^{m-1}}$, based on the maximum margin linear classification to partition the data into two classes. Subsequently, a new test data point~$\widehat{\boldsymbol{\chi}}$ belongs to the class labeled  as~${g(\widehat{\boldsymbol{\chi}})= \text{sgn}(\boldsymbol{\omega}^\top \widehat{\boldsymbol{\chi}} - \nu)}$.  In case the data points are not linearly separable, the input data is first projected into a high-dimensional space~$\mc{F}$ via a nonlinear mapping~$\phi(\cdot)$. This mapping is such that the inner products of two projected data points can be computed via a \textit{kernel function}~$K(\cdot)$, i.e.,~$K(\boldsymbol{\chi}_i,\boldsymbol{\chi}_j)=\phi(\boldsymbol{\chi}_i)^\top \phi(\boldsymbol{\chi}_j)$. By proper selection of~$\phi(\cdot)$, a linear optimal  hyperplane defined by~$\boldsymbol{\omega}$ and~$\nu$ can be found in~$\mc{F}$ such that~${g(\widehat{\boldsymbol{\chi}})= \text{sgn}(\boldsymbol{\omega}^\top \phi(\widehat{\boldsymbol{\chi}}) - \nu)}$ determines the class of~$\widehat{\boldsymbol{\chi}}$. In centralized SVM, all the data points are sent to a central computation entity (the FC) that finds the optimal~$\boldsymbol{\omega}~$ and~$\nu$ by minimizing the following  convex loss function~\cite{chapelle2007training}:
\begin{equation} \label{eq_svm_cent}
	\begin{aligned}
		\displaystyle
		& \min_{\boldsymbol{\omega},\nu}
		~ &  \boldsymbol{\omega}^\top \boldsymbol{\omega} + C \sum_{j=1}^{N} \max\{1-l_j( \boldsymbol{\omega}^\top \phi(\boldsymbol{\chi}_j)+\nu),0\}^q
	\end{aligned}
\end{equation}
where~${q = \{1,2,\ldots\}}$ defines smoothness of the loss function and its derivatives, and the positive constant~$C$ determines the trade-off between increasing the margin size and ensuring that the projected data ~$\phi(\boldsymbol{\chi}_j)$ lies on  correct side of the hyperplane. We note that the SVM loss function~\eqref{eq_svm_cent} is not continuously twice-differentiable for~$q=\{1,2\}$. Therefore, it is common to approximate~$\max\{z,0\}^q$ for~${q=1}$ by the smooth function~${L(z,\mu)=\frac{1}{\mu}\log (1+\exp(\mu z))}$, which is the integral of the well-known \textit{sigmoid function}~\cite{garg2019fixed2}. It can be shown that by setting~$\mu$ large enough~${L(z,\mu)}$ becomes arbitrarily close to~$\max\{z,0\}$; see~\cite{slp_book} for more smooth loss functions, e.g., for logistic regression with cross-entropy loss.

In \textit{distributed} SVM (D-SVM), the data points are available over a network of~$n$ agents and each agent~$i$ possesses a local dataset with~$N_i$ data points denoted  by~$\boldsymbol{\chi}^i_j,{j=1,\ldots,N_i}$.  Since each agent has access to partial data, the locally found values~$\boldsymbol{\omega}_i$ and~$\nu_i$, obtained by solving~\eqref{eq_svm_cent} over the local dataset~$\boldsymbol{\chi}^i_j,{j=1,\ldots,N_i}$,
%$\{\boldsymbol{\chi}^i_j\}_i$, 
may differ for each agent~$i$. The idea behind D-SVM is thus to develop a distributed mechanism to learn the global classifier parameters by making sure that no agent reveals its local data to any other agent. The corresponding distributed optimization  problem is given by: 
\begin{equation} \label{eq_svm_dist}
	\begin{aligned}
		\displaystyle
		 \min_{\boldsymbol{\omega}_1,\nu_1,\ldots,\boldsymbol{\omega}_n,\nu_n}
		\quad &  \sum_{i=1}^{n} f_i(\boldsymbol{\omega}_i,\nu_i) \\
		 \text{subject to} \quad&  \boldsymbol{\omega}_1 = \dots = \boldsymbol{\omega}_n,\qquad\nu_1 = \dots =\nu_n,
	\end{aligned}
\end{equation}
where each local cost~${f_i:\mathbb{R}^m\rightarrow\mbb R}$ is approximated as (with~${z=1-l_j( \boldsymbol{\omega}_i^\top \phi(\boldsymbol{\chi}^i_j)+\nu_i)}$ and large enough~$\mu>0$)
\begin{equation*}
f_i(\boldsymbol{\omega}_i,\nu_i)=\boldsymbol{\omega}_i^\top \boldsymbol{\omega}_i + C \sum_{j=1}^{N_i} \tfrac{1}{\mu}\log (1+\exp(\mu z)).
\end{equation*}
% \begin{equation*}
% f_i(\boldsymbol{\omega}_i,\nu_i)=\boldsymbol{\omega}_i^\top \boldsymbol{\omega}_i + C \sum_{j=1}^{N_i} \max\{1-l_j( \boldsymbol{\omega}_i^\top \phi(\boldsymbol{\chi}^i_j)+\nu_i),0\}^q.
% \end{equation*}
%\begin{rem}\label{rem_strongly convex}
%	The functions~$f_i(\boldsymbol{\omega}_i,\nu_i)$ are strongly convex. 
%\end{rem}
%Problem~\eqref{eq_svm_dist}-\eqref{eq_fi} is a constrained distributed optimization problem. '
Let~${\mb{x}_i = [\boldsymbol{\omega}_i^\top;\nu_i]}\in\mathbb{R}^m$ and let~$\mb{x} \in \mathbb{R}^{mn}$ be the global vector concatenating all~$\mb{x}_i$'s, i.e.,
${\mb{x} = [\mb{x}_1;\mb{x}_2;\dots;\mb{x}_n]}$, where the symbol~`$;$' denotes the column concatenation of the~$\mb{x}_i$ vectors. Then, Problem~\eqref{eq_svm_dist} takes the following form:   
\begin{align}\nonumber
\min_{\mb{x} \in \mathbb{R}^{mn}}
F(\mb{x}),\qquad F(\mb x) = \sum_{i=1}^{n} f_i(\mb{x}_i)\\\label{eq_prob} \text{subject to} ~ \mb{x}_1 = \mb{x}_2 = \dots = \mb{x}_n.
\end{align}
We next provide the following lemma on the local costs. 
% \begin{ass} \label{ass_f(x)}
% 	The local loss functions~$f_i$'s, at all agents, are twice differentiable and strictly convex, i.e., the~$m\times m$ Hessian matrix~$\boldsymbol{ \nabla}^2 f_i(\mb{x}_i)$ is positive definite, for all~$\mb x_i\in\mathbb R^m$.
% \end{ass} 
\begin{lem} \label{lem_f(x)}
\cite{garg2019fixed2} Each local cost~$f_i$ is twice differentiable and strictly convex, i.e., the~$m\times m$ Hessian matrix~$\boldsymbol{ \nabla}^2 f_i(\mb{x}_i)$ is positive definite, for all non-zero~$\mb x_i\in\mathbb R^m$.
\end{lem}

Clearly, any solution~$\mb x_i^*,{i=1,\ldots,n}$, of~\eqref{eq_prob} must satisfy ${\sum_{i=1}^{n} \boldsymbol{ \nabla} f_i(\mb{x}^*_i) = \mb{0}_m}$, such that~${\mb{x}^*_1=\ldots=\mb{x}^*_n=\overline{\mb{x}}^*}$, for some~${\overline{\mb{x}}^*\in\mathbb{R}^m}$. In other words, the optimality condition~${\boldsymbol{ \nabla} F(\mb{x}^*) = \mb{0}_{mn}}$ must hold for some~${\mathbf x^*\in\mathbb{R}^{mn}}$ such that~${\mb{x}^*=\mb{1}_n \otimes \overline{ \mb{x}}^*}$, where ${\boldsymbol{\nabla} F:\mathbb R^{mn}\rightarrow\mathbb R^{mn}}$ is the gradient of~${F:\mbb R^{mn}\rightarrow\mbb R}$.

\section{Proposed Algorithm:\\ Dynamics and Auxiliary Results}\label{sec_dyn}

We now provide a distributed solution to  the D-SVM problem. 
Let~$\mb{x}_i(t)\in\mathbb R^m$ to be the state of agent~$i$ at time~$t$, where~$t\geq 0$ is the continuous-time variable. To solve problem~\eqref{eq_prob}, we consider the following continuous-time linear dynamics for all~$\mb{x}_i(t) \in \mathbb{R}^{m},i\in\{1,\ldots,n\}$, 
\begin{eqnarray} \label{eq_xdot}	 
	\dot{\mb{x}}_i = -\sum_{j=1}^{n} w_{ij}(\mb{x}_i-\mb{x}_j)-\alpha \mb{y}_i, 
\end{eqnarray}
where~${\dot{\mb{x}}_i=\frac{d{\mb{x}}_i}{dt}}$,~${W=\{w_{ij}\}}$ is the weighted adjacency matrix associated with ~$\mc{G}$, and~$\alpha>0$ is the stepsize.  We note that instead of the standard descend direction of~$\nabla f_i(\mathbf{x}_i)$, the~$\mathbf{x}_i$-update in~\eqref{eq_xdot}
descends in the direction of an auxiliary variable~$\mb y_i(t)\in\mathbb{R}^m$. The variable~$\mb y_i(t)$ in fact tracks the sum of local gradients, asymptotically, and is updated via the following dynamics (see~\cite{xin2020decentralized,xin2020general,giannakis2020time} for similar DT methods):
\begin{eqnarray} \label{eq_ydot}	
	\dot{\mb{y}}_i &=& -\sum_{j=1}^{n} a_{ij}(\mb{y}_i-\mb{y}_j) + \frac{d}{d t} \boldsymbol{ \nabla} f_i(\mb{x}_i),
\end{eqnarray}
where~$\dot{\mb{y}}_i=\frac{d{\mb{y}}_i}{dt}$ and the matrix~$A=\{a_{ij}\}$ is the weighted adjacency matrix with the same structure as~$W$. In~\eqref{eq_ydot},
\begin{eqnarray} \label{eq_dtdf}	
	\frac{d}{d t} \boldsymbol{ \nabla} f_i(\mb{x}_i)=  \boldsymbol{ \nabla}^2 f_i(\mb{x}_i)  \dot{\mb{x}}_i.
\end{eqnarray} 
Note that the proposed algorithm, \eqref{eq_xdot} and \eqref{eq_ydot}, is in continuous-time. However, the structure of the underlying graph~$\mathcal G$ may change in time instances, that we consider as time steps in a discrete-time framework. This makes the proposed dynamics \textit{hybrid} where the state variables,~$\mathbf x$ and~$\mathbf y$, evolve in CT over DT switching of the network topology.
For the ease of notation, we define an auxiliary global variable~${\mb{y}=[\mb{y}_1;\mb{y}_2;\dots;\mb{y}_n]\in\mbb{R}^{mn}}$ that concatenates the local~$\mb{y}_i(t)$'s. %As introduced in Section~\ref{sec_pre}, the associated Laplacian matrices for~$W$ and~$A$, are denoted by~$\overline{W}$ and~$\overline{A}$, respectively, satisfying the following assumption. 
We make the following assumption on the weight matrices $W$ and $A$.
%; recall Section~\ref{zz}.
\begin{ass} \label{ass_Wg}
		The weights~${W=\{w_{ij}\}}$ and~${A=\{a_{ij}\}}$ are associated to a WB-digraph with~${w_{ij}\geq0}$ and~${a_{ij}\geq0}$, respectively. 
		%for~$i,j \in \{1,\ldots,n\}$.
		Further,~$\sum_{j=1}^n w_{ij}<1$ and ~${\sum_{j=1}^n a_{ij}<1}$.
\end{ass}
Following Assumption~\ref{ass_Wg}, we obtain from~\eqref{eq_xdot} and~\eqref{eq_ydot}: 
\begin{eqnarray}
	\sum_{i=1}^n \dot{\mb{y}}_i 
% 	&=& -\sum_{i,j=1}^{n} a_{ij}(\mb{y}_i-\mb{y}_j) + \sum_{i=1}^n\frac{d}{d t} \boldsymbol{ \nabla} f_i(\mb{x}_i)\\	
	&=& \sum_{i=1}^n\frac{d}{d t} \boldsymbol{ \nabla} f_i(\mb{x}_i),  \label{eq_sumydot} \\ \label{eq_sumxdot} 
	\sum_{i=1}^n \dot{\mb{x}}_i 
% 	&=& -\sum_{i,j=1}^{n} w_{ij}(\mb{x}_i-\mb{x}_j)-\alpha \sum_{i=1}^n\mb{y}_i \\ 
	&=& -\alpha \sum_{i=1}^n\mb{y}_i.
\end{eqnarray}
Integrating \eqref{eq_sumydot} with respect to~$t$ and initializing the auxiliary variable~$\mb{y}(0)=\mb{0}_{nm}$, we have
\begin{eqnarray} \label{eq_sumxdot2}
	\sum_{i=1}^n \dot{\mb{x}}_i = -\alpha \sum_{i=1}^n\mb{y}_i = -\alpha \sum_{i=1}^n \boldsymbol{ \nabla} f_i(\mb{x}_i),
\end{eqnarray}
which shows that the time-derivative of the sum of states~$\mb{x}_i$'s is towards sum gradient, and therefore, the equilibrium ($\dot {\mathbf x_i}=0_m$) of the dynamics~\eqref{eq_xdot}-\eqref{eq_ydot} is~$\mb{x}^*$ satisfying~${(\mathbf 1_n^\top \otimes I_m) \boldsymbol{ \nabla} F(\mb{x}^*) = \mb{0}_m}$ ($I_m$ as the identity matrix of size $m$), which is the optimal state of  problem~\eqref{eq_prob}~\cite{gharesifard2013distributed}. 
%This is stated in the following lemma.
\begin{lem} \label{lem_invariant}
	Initializing from any~$\mb{x}(0) \neq \mb{1}_n \otimes  {\mb{x}}_0$, for some non-zero~$\mathbf{x}_0\in\mathbb{R}^m$, and~$\mb{y}(0)=\mb{0}_{nm}$, the  state~$[\mb{x}^*;\mb{0}_{nm}]$ with $(\mathbf 1_n^\top \otimes I_m) \boldsymbol{ \nabla} F(\mb{x}^*) = \mb{0}_m$ is an invariant equilibrium point of the dynamics~\eqref{eq_xdot}-\eqref{eq_ydot}.
\end{lem}
\begin{proof}
	From~\eqref{eq_sumxdot2}, 
	%and due to the strict convexity of~$F(\mb{x})$, 
	the following uniquely holds at~${\mb{x}=\mb{x}^*=\mb{1}_n \otimes \overline{ \mb{x}}^*}$,
	\[\sum_{i=1}^n \dot{\mb{x}}_i = -\alpha (\mathbf 1_n^\top \otimes I_m) \boldsymbol{ \nabla} F(\mb{x}^*) =  \mb{0}_m.
	\]
	Further, from~\eqref{eq_xdot} we have~$\dot{\mb{x}}_i = \mb{0}_m$
	and from~\eqref{eq_ydot} and~\eqref{eq_dtdf},
	\[\dot{\mb{y}}_i = \frac{d}{d t} \boldsymbol{ \nabla} f_i(\overline{ \mb{x}}^*)=  \boldsymbol{ \nabla}^2 f_i(\overline{ \mb{x}}^*) \dot{\mb{x}}_i = \mb{0}_m,
	\] 
    which shows that~$[\mb{x}^*;\mb{0}_{nm}]$ is an invariant equilibrium point of the dynamics~\eqref{eq_xdot}-\eqref{eq_ydot}. 
\end{proof}
The above lemma only shows that the state~$[\mb{x}^*;\mb{0}_{nm}]$, with~$\mb{x}^*$ as the optimal point of problem~\eqref{eq_prob}, is the equilibrium of the proposed networked dynamics~\eqref{eq_xdot}-\eqref{eq_ydot}. 
%The detailed proof of convergence, i.e., that~\eqref{eq_xdot}-\eqref{eq_ydot} converge to this equilibrium, is given in the next section. 
Note that the first term in Eq.~\eqref{eq_xdot} drives the agents to reach consensus on~$\mb{x}_i$'s while the second term along with the dynamics~\eqref{eq_ydot} implements the gradient correction. 
%Therefore, Eq.~\eqref{eq_xdot}-\eqref{eq_ydot} can be perceived as a constrained gradient tracking dynamics.

%\begin{rem}
%{\uk Unlike non-Lipschitz solutions~\cite{ning2017distributed,garg2019fixed2,rahili_ren,taes2020finite},  the proposed dynamics~\eqref{eq_xdot}-\eqref{eq_ydot} are Lipschitz continuous, which results in chattering-free DT approximation.}
%\end{rem}
% \begin{rem}
%	Unlike~\cite{gharesifard2013distributed}, which requires ${\sqrt{3}|\operatorname{Im}\{\lambda\}| \leq |\operatorname{Re}\{\lambda\}|}$ with~$\lambda$ as the eigenvalues of Laplacian matrix, we assume no such restriction.  	%on ~$\overline{W}$ and~$\overline{A}$.
%	Moreover, this work considers possibly switching WB-digraphs while~\cite{gharesifard2013distributed} is restricted to time-invariant networks.
%\end{rem}

\section{Proof of Convergence} \label{sec_conv}
In this section, we show that  dynamics~\eqref{eq_xdot}-\eqref{eq_ydot} converge to the equilibrium state described in Lemma~\ref{lem_invariant}.
As it is the case in Section~\ref{sec_pre}, in which the Laplacian matrix for~$W$ is $\overline{W}$, the Laplacian matrix for~$A$ is denoted by~$\overline{A}$. 
Define the~$nm$-by-$nm$  Hessian matrix~$H\coloneqq\text{blockdiag}[\boldsymbol{ \nabla}^2 f_i(\mb{x}_i)]$.
The dynamics~\eqref{eq_xdot}-\eqref{eq_ydot} can be written in compact form as
\begin{eqnarray} \label{eq_xydot}
	\left(\begin{array}{c} \dot{\mb{x}} \\ \dot{\mb{y}} \end{array} \right) = M (t,\alpha) \left(\begin{array}{c} {\mb{x}} \\ {\mb{y}} \end{array} \right),
\end{eqnarray} 
\begin{eqnarray} \label{eq_M}
M(t,\alpha) = \left(\begin{array}{cc} \overline{W} \otimes I_m & -\alpha I_{mn} \\ H(\overline{W}\otimes I_m) & \overline{A} \otimes I_m - \alpha H
\end{array} \right).
\end{eqnarray}
Networked dynamics~\eqref{eq_xydot}-\eqref{eq_M} represents a \textit{hybrid dynamical system} because: (i) the matrix~$H$ varies in CT; and (ii) the structure of $\overline{W}$ and~$\overline{A}$ may change in DT in case of \textit{dynamic network topology}. In this direction, towards convergence analysis, we first evaluate the stability properties of the matrix~$M$ at every time-instant, and then generalize the convergence to the entire time horizon. In the rest of this paper for notation simplicity, we drop the dependence of~$M$ on~$(t,\alpha)$, unless where needed, despite the fact that it is a function of both time~$t$ and stepsize~$\alpha$. 

\begin{lem} \label{lem_dM} {\hspace{-0.005cm}~\cite{seyranian2003multiparameter,cai2012average}}
Let~$P(\alpha)$ be an~$n$-by-$n$ matrix depending smoothly on a real parameter~$\alpha \geq 0$. Assume~$P(0)$ has~$l<n$ equal eigenvalues, denoted by~$\lambda_1=\ldots=\lambda_l$, associated with right eigenvectors~$\mb{v}_1,\ldots,\mb{v}_l$ and left eigenvectors~$\mb{u}_1,\ldots,\mb{u}_l$ which are linearly independent. 
Let~$\lambda_i(\alpha)$ denote the eigenvalues of~$P(\alpha)$, as a function of~$\alpha$, corresponding to~${\lambda_i, i \in \{1,\ldots,l\}}$, and~$P' = \frac{dP(\alpha)}{d\alpha}|_{\alpha=0}$. Then,~$\frac{d\lambda_i}{d\alpha}|_{\alpha=0}$ are the eigenvalues of the following~$l$-by-$l$ matrix,
\[\left(\begin{array}{ccc}
\mb{u}_1^\top P' \mb{v}_1 & \ldots & \mb{u}_1^\top P' \mb{v}_l \\
 & \ddots & \\
\mb{u}_l^\top P' \mb{v}_1 & \ldots & \mb{u}_l^\top P' \mb{v}_l 
\end{array} \right).
\]  
\end{lem}

\begin{thm} \label{thm_zeroeig}
Let Assumption~\ref{ass_Wg} hold. For a sufficiently small~$\alpha$, all eigenvalues of~$M$ have non-positive real-parts and the algebraic multiplicity of the zero eigenvalue is~$m$. 
\end{thm}
\begin{proof}
Let~$M=M_0+\alpha M_1$ with 
\begin{eqnarray}\nonumber
	M_0 &=&  \left(\begin{array}{cc} \overline{W} \otimes I_m & \mb{0}_{mn\times mn} \\ H(\overline{W} \otimes I_m) & \overline{A}\otimes I_m \end{array} \right),\\\nonumber
	M_1 &=& \left(\begin{array}{cc} \mb{0}_{mn\times mn} & - {I_{mn}} \\ {\mb{0}_{mn\times mn}} & - H \end{array} \right),
\end{eqnarray}
where~$\mb{0}_{mn\times mn}$ is the zero matrix of size~$mn$.
Since matrix~$M_0$ is block (lower) triangular we have,
\begin{eqnarray}
	\sigma(M_0) = \sigma(\overline{W} \otimes I_m) \cup \sigma(\overline{A} \otimes I_m),
\end{eqnarray}
where~$\sigma(\cdot)$ represents the eigenspectrum of the matrix. From Lemma~\ref{lem_sc}, both matrices ~$\overline{W}$ and~$\overline{A}$ have~$n-1$ eigenvalues in the LHP (left-half plane) and one isolated eigenvalue at zero. Therefore,  matrix~$M_0$ has~$m$ sets of eigenvalues associated with~$m$ dimensions of vector states~$\mb{x}_i$ i.e.,
$$\operatorname{Re}\{\lambda_{2n,j}\} \leq \ldots \leq \operatorname{Re}\{\lambda_{3,j}\} < \lambda_{2,j} = \lambda_{1,j} = 0,$$ where $j=\{1,\ldots,m\}.$
Using Lemma~\ref{lem_dM}, we analyze  the spectrum of~$M$ by considering it as the perturbed version of~$M_0$ via the term~$\alpha M_1$. We check the variation of the zero eigenvalues~$\lambda_{1,j}$ and~$\lambda_{2,j}$ by adding the (small) perturbation~$\alpha M_1$. Denote these perturbed eigenvalues by~$\lambda_{1,j}(\alpha)$ and~$\lambda_{2,j}(\alpha)$. To apply Lemma~\ref{lem_dM}, define the right eigenvectors corresponding to~$\lambda_{1,j}$ and~$\lambda_{2,j}$ as,
%    \[V = (\mb{v}_{1,1} \ldots \mb{v}_{1,m} ~\mb{v}_{2,1} \ldots \mb{v}_{2,m}) =\left(\begin{array}{cc}
%    \mb{1}_n& \mb{0}_n \\
%	\mb{0}_n & \mb{1}_n 
%    \end{array} \right)\otimes I_m \]
\begin{equation} \label{eq_V}
V = [V_1~ V_2] =\left(\begin{array}{cc}
	\mb{1}_n& \mb{0}_n \\
	\mb{0}_n & \mb{1}_n 
\end{array} \right)\otimes I_m,
\end{equation}
%\[ U^\top =\left(\begin{array}{c}
%	\mb{u}_{1,1}^\top\\ \vdots \\ \mb{u}_{1,m}^\top \\ \mb{u}_{2,1}^\top \\ \vdots \\ \mb{u}_{2,m}^\top
%\end{array}
%\right)=\left(\begin{array}{cc}
%	\mb{1}_n^\top	& \mb{0}_n^\top \\
%	\mb{0}_n^\top	&  \mb{1}_n^\top 
%\end{array} \right) \otimes I_m
%\]
    Similarly, the left eigenvectors are~$V^\top$. Note that these eigenvectors are defined using Lemma~\ref{lem_laplacian} and satisfy~$V^\top V=I_{2mn}$. Recall that, ~$\frac{dM(\alpha)}{d\alpha}|_{\alpha=0}=M_1$ and following Lemma~\ref{lem_dM},
    	\begin{eqnarray} \label{eq_dmalpha}
    		V^\top M_1 V= \left(\begin{array}{cc}
    \mb{0}_{m\times m}	& \mb{0}_{m\times m} \\
    -nI_m	& -(\mb{1}_n \otimes I_m)^\top H (\mb{1}_n \otimes I_m)
    \end{array} \right).
    \end{eqnarray}  
    Following the definition of the Hessian matrix~$H$,
    \begin{equation} \label{eq_sum_df}
    	-(\mb{1}_n \otimes I_m)^\top H (\mb{1}_n \otimes I_m)= -\sum_{i=1}^n \boldsymbol{ \nabla}^2 f_i(\mb{x}_i) \prec 0,
    \end{equation}
    where the last inequality follows the strict convexity of the local functions~$f_i(\mb{x}_i)$ (Lemma~\ref{lem_f(x)}). 
    %Note that only for~$\mb{x}_i=\overline{\mb{x}}^*$ we have~$\sum_{i=1}^n \boldsymbol{ \nabla}^2 f_i(\mb{x}_i) = \mb{0}_{m\times m}$. 
    Recall that from Lemma~\ref{lem_dM} the derivatives~$\frac{d\lambda_{1,j}}{d\alpha}|_{\alpha=0}$ and~$\frac{d\lambda_{2,j}}{d\alpha}|_{\alpha=0}$ depend on the eigenvalues of~\eqref{eq_dmalpha} which clearly form a lower triangular matrix with~$m$ zero eigenvalues and~$m$ negative eigenvalues (following~\eqref{eq_sum_df}). Therefore,~$\frac{d\lambda_{1,j}}{d\alpha}|_{\alpha=0} = 0$ and~$\frac{d\lambda_{2,j}}{d\alpha}|_{\alpha=0}<0$, 
    %(for~$\mb{x}_i \neq \overline{\mb{x}}^*$) 
    which implies that considering~$\alpha M_1$ as a perturbation, the~$m$ zero eigenvalues~$\lambda_{2,j}(\alpha)$ of~$M$ move toward the LHP while~$\lambda_{1,j}(\alpha)$'s remain zero. We recall that the eigenvalues are a continuous functions of the matrix elements~\cite{stewart1990matrix}, and therefore, for sufficiently small~$\alpha$ we have,    
    \begin{equation}
    \begin{aligned}
         \operatorname{Re}\{\lambda_{2n,j}(\alpha)\} &\leq \ldots \leq \operatorname{Re}\{\lambda_{3,j}(\alpha)\} \\
        &\leq \lambda_{2,j}(\alpha) < \lambda_{1,j}(\alpha) = 0.
    \end{aligned}  
    \end{equation}
The proof is completed.
\end{proof}
Theorem~\ref{thm_zeroeig} proves that for sufficiently small~$\alpha$ the matrix~$M$ defined by~\eqref{eq_M} has~$m$ zero eigenvalue each associated with a dimension of variable~$\mb{x}_i$,  while all other eigenvalues remain in the LHP. Next, we determine an  upper-bound on~$\alpha$ guaranteeing  the results of Theorem~\ref{thm_zeroeig}. In this direction,  we provide some relevant concepts regarding~$\sigma(M_0)$ and~$\sigma(M)$ as the spectrum of~$M_0$ and~$M$ matrices. Define the optimal matching distance as~\cite{bhatia2013matrix}, \[d(\sigma(M),\sigma(M_0)) = \min_{\pi} \max_{1\leq i\leq 2nm} (\lambda_i - \lambda_{\pi(i)}(\alpha)),
\]
with~$\pi$ representing all possible permutations over the set~$\{1,\ldots,2nm\}$. In fact,~$d(\sigma(M),\sigma(M_0))$ is the smallest-radius circle among the circles centered at~$\lambda_{1,j},\ldots,\lambda_{2n,j}$ (eigenvalues of~$M_0$) which includes all the eigenvalues of~$M$ denoted by~$\lambda_{1,j}(\alpha),\ldots,\lambda_{2n,j}(\alpha)$. Loosely speaking,~$d(\sigma(M),\sigma(M_0))$ represents the furthest distance between the eigenvalues of matrices~$M$ and~$M_0$. Recall that from Theorem~\ref{thm_zeroeig} we know that the first~$2m$ eigenvalues of the perturbed matrix~$M$ are~$\lambda_{1,j}(\alpha)=0$ and~$\lambda_{2,j}(\alpha)<0$. To show that all the other~$(2n-2)m$ eigenvalues~$\lambda_{3,j}(\alpha),\ldots,\lambda_{2n,j}(\alpha)$ remain in the LHP, it is sufficient  that,
\begin{eqnarray} \label{eq_lambda3}
	d(\sigma(M),\sigma(M_0))<\lambda_{\min},
\end{eqnarray}
where~$\lambda_{\min}=\min_{1\leq j\leq m}|\operatorname{Re}\{\lambda_{3,j}\}|$. This guarantees that the distance between the~$(2n-2)m$ eigenvalues of~$M_0$ and~$M$ is less than~$\lambda_{\min}$ and therefore all the~$(2n-2)m$ eigenvalues of~$M$ remain in the LHP. In this direction, the following lemma provides a useful bound on~$d(\sigma(M),\sigma(M_0))$. 
\begin{lem} \label{eq_dbound}\cite{bhatia2013matrix}
The following holds for~${M = M_0 +\alpha M_1}$:
\begin{eqnarray*}
d(\sigma(M),\sigma(M_0))\leq 4(\lVert M_0\rVert_{\infty}+\lVert M\rVert_{\infty})^{1-\frac{1}{nm}} \lVert \alpha M_1\rVert_\infty^{\frac{1}{nm}}.
\end{eqnarray*}
\end{lem} 
From~\eqref{eq_lambda3} and Lemma~\ref{eq_dbound}, $\alpha$ can be bounded as follows.
\begin{lem} \label{lem_alphabar}
Define~${\gamma = \max_{1\leq i\leq nm} \sum_{j=1}^{nm} |H_{ij}|}$ and ${\lambda_{\min}=\min_{1\leq j\leq m}|\operatorname{Re}\{\lambda_{3,j}\}|}$. The real-part of the eigenvalues~${\operatorname{Re}\{\lambda_{3,j}(\alpha)\},\ldots,\operatorname{Re}\{\lambda_{2n,j}(\alpha)\}<0}$ if~$0<\alpha<\overline{\alpha}$ where for~$\gamma<1$,
\begin{equation} \label{eq_alphabar1}
\overline{\alpha} = \argmin_{\alpha>0} |4(\max\{4+4\gamma+\alpha\gamma, 4+2\gamma+\alpha\})^{1-\frac{1}{nm}} \alpha^{\frac{1}{nm}}-\lambda_{\min}|,
\end{equation}	
and for~$\gamma\geq 1$,		
\begin{equation} \label{eq_alphabar2}
\overline{\alpha} = \argmin_{\alpha>0} |4(4+4\gamma+\alpha\gamma)^{1-\frac{1}{nm}} (\alpha \gamma)^{\frac{1}{nm}}-\lambda_{\min}|.
\end{equation}
\end{lem}   
\begin{proof}
From Assumption~\ref{ass_Wg} and  Lemmas~\ref{lem_laplacian}-\ref{lem_f(x)},$$||M_0||_\infty\leq 2(1+\gamma).$$
This is because, from Assumption~\ref{ass_Wg} the row sum of the absolute values of matrix~$\overline{W}$ is at most~$2$ and similarly, using Lemma~\ref{lem_f(x)}, the row sum of the absolute values of matrix~$H(\overline{W} \otimes I_m)$ is at most~$2\gamma$. By similar reasoning, 
\[||M||_\infty\leq \max\{2+\gamma(2+\alpha), 2+\alpha\},
\]
\[||\alpha M_1||_\infty\leq \max\{\alpha \gamma, \alpha\}.\]
%	Assuming~$\alpha<1$,
%	\[||M||_\infty\leq 3+\gamma(1+\alpha).
%	\]
Following~\eqref{eq_lambda3}, for~$\gamma<1$,
\[ 4(2(1+\gamma)+\max\{2+\gamma(2+\alpha), 2+\alpha\})^{1-\frac{1}{nm}} \alpha^{\frac{1}{nm}}<\lambda_{\min},
	\]
and for ~$\gamma\geq1$,
\[ 4(4+\gamma(4+\alpha))^{1-\frac{1}{nm}} (\alpha \gamma)^{\frac{1}{nm}}<\lambda_{\min}.
\]		
Since the functions on the left-hand-side of the above inequalities are monotonically increasing for~$\alpha>0$, the largest~$\alpha$ satisfying the above inequalities is given  by~\eqref{eq_alphabar1}-\eqref{eq_alphabar2}.	
\end{proof}	
Lemma~\ref{lem_alphabar} gives a conservative upper-bound on~$\alpha$ which guarantees the rest of the eigenvalues, other than~$\lambda_{1,j}(\alpha)=0$ and~$\lambda_{2,j}(\alpha)<0$, remain in the LHP and Theorem~\ref{thm_zeroeig} is valid.	However, the eigenvalues of~$M$ may still remain in the LHP for a possible less-conservative choice of~$\alpha>\overline{\alpha}$. In general, for a proper~$\alpha$, matrix~$M$ has~$m$ zero eigenvalues associated with the eigenvectors~$V_1$ given in~\eqref{eq_V}, and the \textit{null space} of the \textit{time-varying} matrix~$M$ is,
\begin{equation}
	\mc{N}(M) = \text{span}\{\left(\begin{array}{c}
		\mb{1}_n \\
		\mb{0}_n  
	\end{array} \right)\otimes I_m\},
\end{equation}
which is \textit{independent of time}.

\begin{thm} \label{thm_lyapunov}
Let the conditions in Lemma~\ref{lem_invariant}, Lemma~\ref{lem_alphabar}, and Theorem~\ref{thm_zeroeig} hold.  The proposed dynamics~\eqref{eq_xdot}-\eqref{eq_ydot} converges to~$[\mb{x}^*;\mb{0}_{nm}]$ with~$\mb{x}^*$ as the optimal value of problem~\eqref{eq_prob}. 
\end{thm}
\begin{proof}
Define the following proper positive-definite Lyapunov function proposed in~\cite{mesbahi2010graph},
\begin{eqnarray}
	V(\boldsymbol{\delta}) = \frac{1}{2} \boldsymbol{\delta}^\top \boldsymbol{\delta} =  \frac{1}{2}\lVert\boldsymbol{\delta}\rVert_2^2
\end{eqnarray}
with~$\delta \in \mathbb{R}^{2mn}$ defined as the difference of system state and the optimal state,
\begin{eqnarray}
	\boldsymbol{\delta} = \left(\begin{array}{c} {\mb{x}} \\ {\mb{y}} \end{array} \right)-\left(\begin{array}{c} {\mb{x}^*} \\ {\mb{0}_{mn}} \end{array} \right).
\end{eqnarray}
Since from Lemma~\ref{lem_invariant}~$[\mb{x}^*;\mb{0}_{nm}]$ is an invariant state of the networked dynamics~\eqref{eq_xydot}-\eqref{eq_M}, we have
$\dot{\boldsymbol{\delta}} = M \boldsymbol{\delta}$.
Then, the time-derivative of the proposed Lyapunov function is as follows,
\[\dot{V} = \boldsymbol{\delta}^\top \dot{\boldsymbol{\delta}}=  \boldsymbol{\delta}^\top M \boldsymbol{\delta},
%= 
%\boldsymbol{\delta}^\top M^\top \boldsymbol{\delta}=
%\boldsymbol{\delta}^\top \frac{M+M^\top}{2} \boldsymbol{\delta}
\]
Following Theorem~\ref{thm_zeroeig}, let~${\lambda}_{1,j}(\alpha)=0$,~$Re\{{\lambda}_{i,j}(\alpha)\}<0$ for~$2\leq i\leq2n,1\leq j\leq m$ represent the real-parts of the eigenspectrum of~$M$. 
%Let~$\widehat{M}=\frac{M+M^\top}{2}$ represent the symmetric part of matrix~$M$. Following Theorem~\ref{thm_zeroeig} let~$\lambda_1=0$,~$Re\{\lambda_i\}<0$ for~$2\leq i\leq2n$ represent the spectrum of matrix~$M$. It can be proved that since matrix~$\widehat{M}$ is symmetric all its eigenvalues are real  ~$\widehat{\lambda}_1=0$,~$\widehat{\lambda_i}<0$ for~$2\leq i\leq2n$~\cite{hornjohnson}. 
It is known that~\cite{SensNets:Olfati04},
\begin{eqnarray} \label{eq_Re2}
	\boldsymbol{\delta}^\top M \boldsymbol{\delta} \leq \max_{1\leq j\leq m}\operatorname{Re}\{{\lambda}_{2,j}(\alpha)\} \boldsymbol{\delta}^\top  \boldsymbol{\delta}. 
\end{eqnarray}
Since~$M$ varies in time,~$\max_{1\leq j\leq m}\operatorname{Re}\{{\lambda}_{2,j}(\alpha)\}$ also changes in time. However, from Theorem~\ref{thm_zeroeig}, it is always negative, implying that~$\dot{V}< 0$ for~$\delta \neq \mb{0}_{2mn}$. 
%Recall that for~$\delta = \mb{0}_{2mn}$ we have~$\sum_{i=1}^n \boldsymbol{ \nabla}^2 f_i(\mb{x}_i) = \mb{0}_{m\times m}~$ and, following~\eqref{eq_sum_df},~${\lambda}_{2,j}(\alpha)=0$. 
We have,
\[\dot{V} =  0 \Leftrightarrow \boldsymbol{\delta} = \mb{0}_{2mn},
\] 
and, from LaSalle’s invariance principle, convergence to the invariant set~$\{\boldsymbol{\delta} = \mb{0}_{2mn}\}$ follows (see \cite{mesbahi2010graph} Section 4.1).
\end{proof}

\begin{rem}
 	Following~\eqref{eq_Re2}, the convergence rate of the dynamics~\eqref{eq_xydot}-\eqref{eq_M} depends on~$\operatorname{Re}\{{\lambda}_{2,j}(\alpha)\}$. Note that~$\operatorname{Re}\{{\lambda}_{2,j}(\alpha)\}$ is tightly related to the parameter~$\alpha$ and therefore, to improve the convergence rate parameter~$\alpha$ needs not to be very small. 
\end{rem}

\section{Simulation: Nonlinear SVM Example} \label{sec_sim}
For simulation consider the academic example given in~\cite{russell2010artificial} (page~$747$). Consider~$N=60$ uniformly distributed sample data points, shown in Fig.~\ref{fig_data}(Left), represented in two classes:  blue `*'s and  red `o's. Clearly, these points~$\boldsymbol{\chi}_i=[{\chi}_i(1);{\chi}_i(2)]$ are not linearly separable in~$\mathbb{R}^2$. The nonlinear mapping~$\phi(\boldsymbol{\chi}_i) = [{\chi}_i(1)^2;{\chi}_i(2)^2;\sqrt{2}{\chi}_i(1){\chi}_i(2)]$, proposed by~\cite{russell2010artificial}, properly maps the data points to~$\mathbb{R}^3$ such that the projected data points are linearly separable via a hyperplane as shown in Fig.~\ref{fig_data}(Right). It can be shown that  the associated kernel function is
$K(\boldsymbol{\chi}_i,\boldsymbol{\chi}_j)=(\phi(\boldsymbol{\chi}_i)^\top \phi(\boldsymbol{\chi}_j))^2$. 
\begin{figure}[]
	\centering
	\includegraphics[width=1.4in]{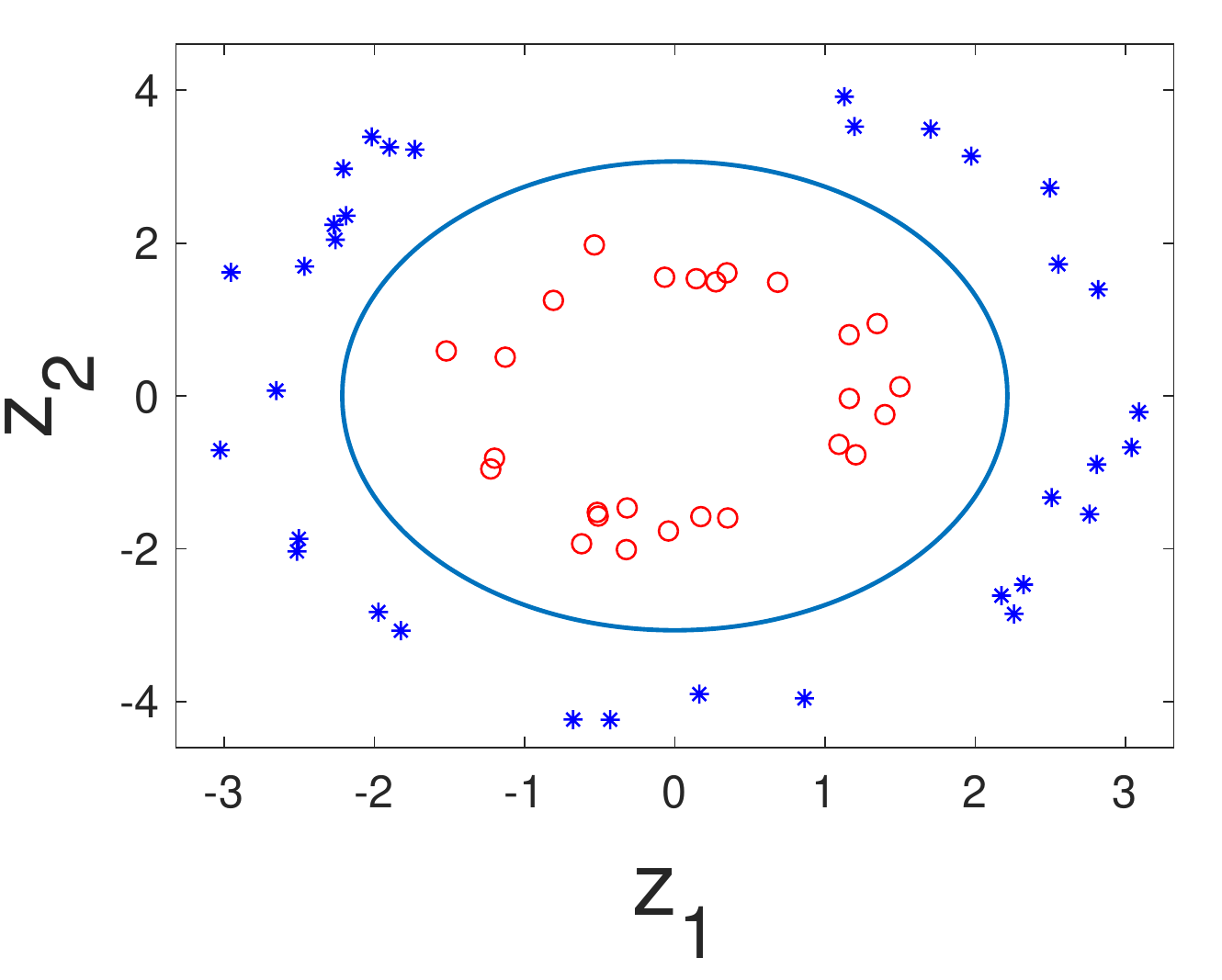}
	\hspace{0.2cm}
	\includegraphics[width=1.6in]{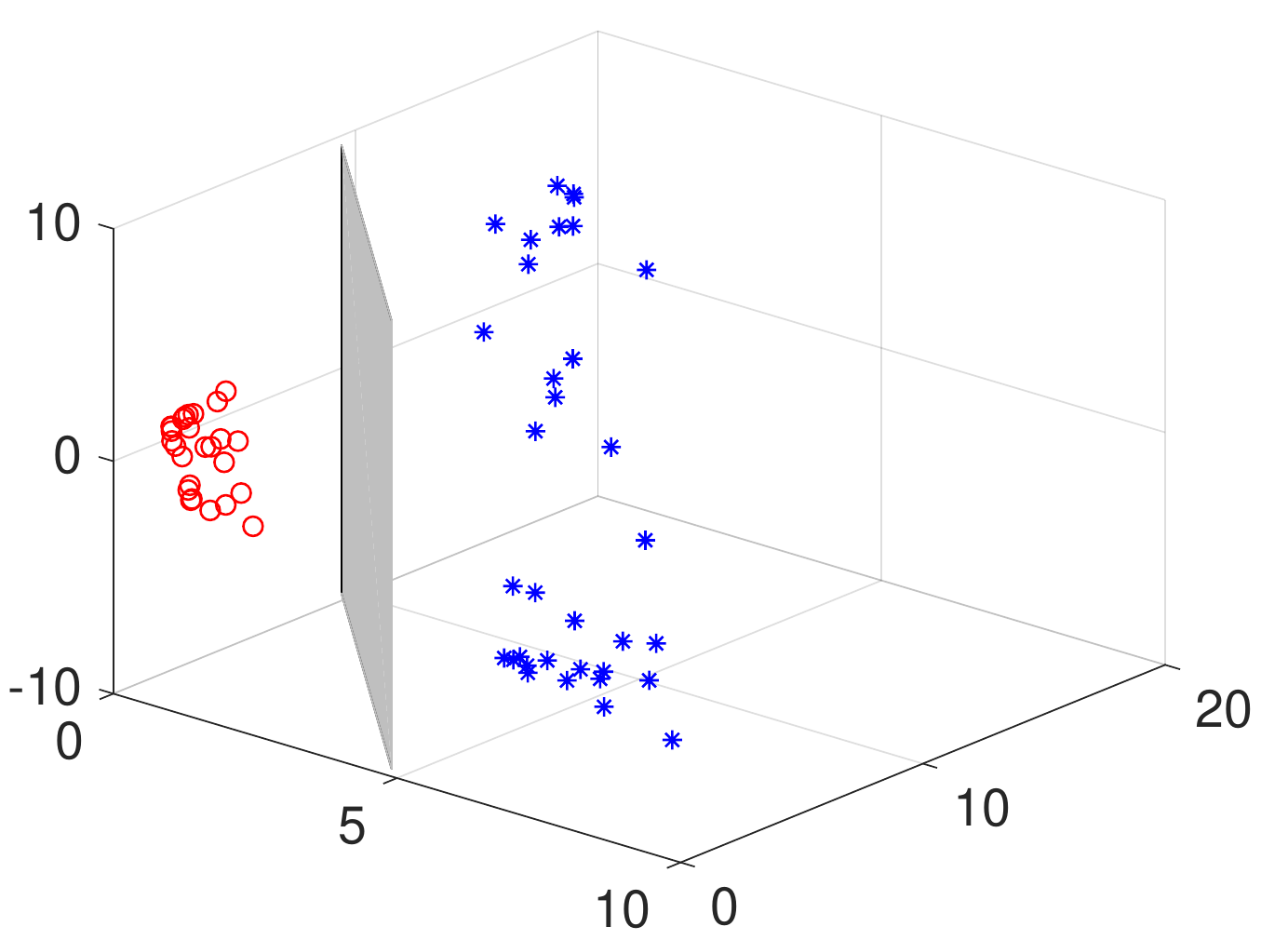}	
	\caption{(Left) Training data and the optimal nonlinear classifier (the ellipse) in 2D. (Right) The same points mapped into 3D space via a nonlinear mapping. Linear SVM optimally classifies the data points via the gray hyperplane which represents the  ellipse in the left figure by inverse mapping.     }
	\label{fig_data}
\end{figure}
We evaluate the proposed dynamics~\eqref{eq_xydot}-\eqref{eq_M} (with $\alpha=10$) for D-SVM over a network of~$n=5$ agents each having access to~$50\%$ random selection of the data points. 
The loss function~$f_i(\mb{x}_i)$ follows the smooth approximation discussed in Section~\ref{sec_prob} with~$\mu=3$ and~$C=1.5$. Every agent finds the optimal separating hyperplane defined by~$\mb{x}_i = [\boldsymbol{\omega}_i^\top;\nu_i]$ ($\boldsymbol{\omega}_i \in \mathbb{R}^3$) and shares this value along with the auxiliary variable~$\mb{y}_i$ with its direct neighbors in $\mc{G}$. The agents' network $\mc{G}$ is considered as the union of a directed cycle and a $2$-hop digraph
(see examples in~\cite{xin2020general}). To satisfy the weight-balanced condition the link weights in each network are equal and randomly chosen in the range~$(0,0.5)$. Using MATLAB's \texttt{randperm} function, we randomly change the permutation of the nodes in the network and the link weights every~$0.05$ seconds to simulate  a dynamic network in DT domain. The time-evolution of~$\mb{x}_i  = [\boldsymbol{\omega}_i^\top;\nu_i] \in \mathbb{R}^4$, loss function~$F(\mb{x})$, and sum of the gradients~$\sum_{i=1}^{5} \boldsymbol{ \nabla} f_i(\mb{x}_i) \in \mathbb{R}^4$ are shown in  Fig.~\ref{fig_dynamics}. 
\begin{figure}
	\centering
	\includegraphics[width=3.5in]{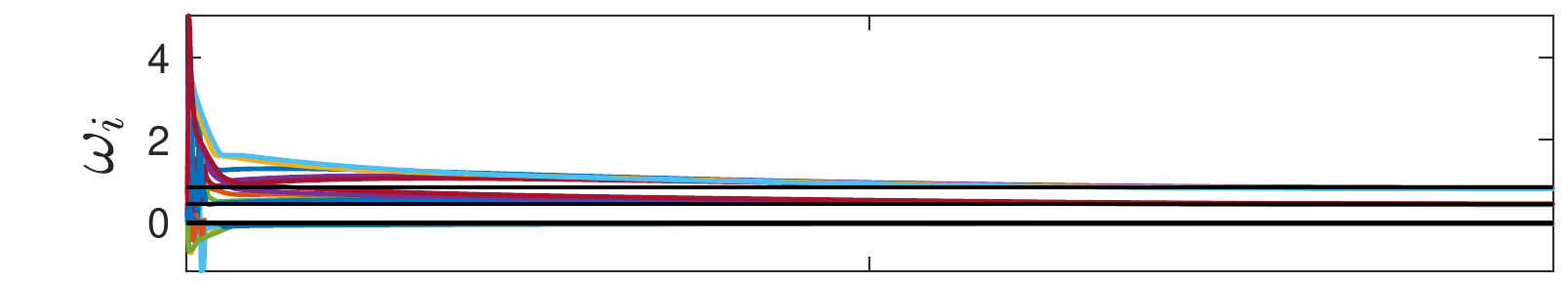}
	\includegraphics[width=3.51in]{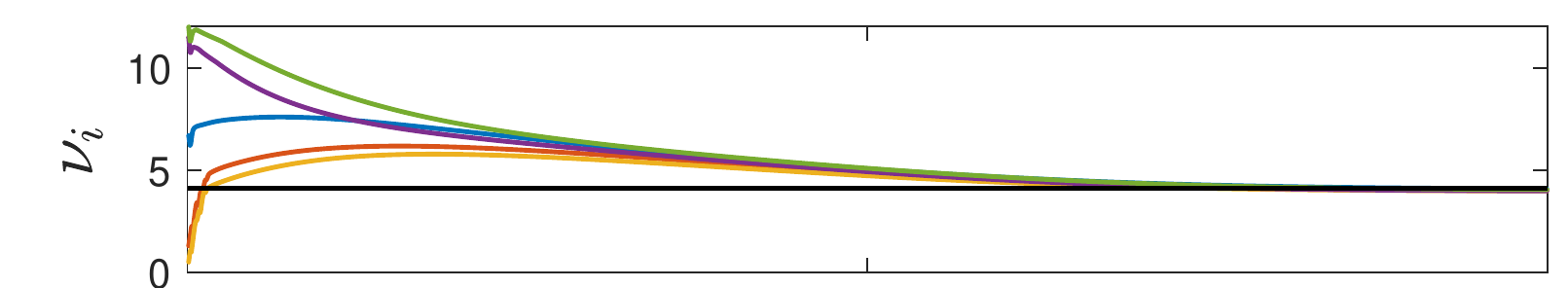}	
	\includegraphics[width=3.5in]{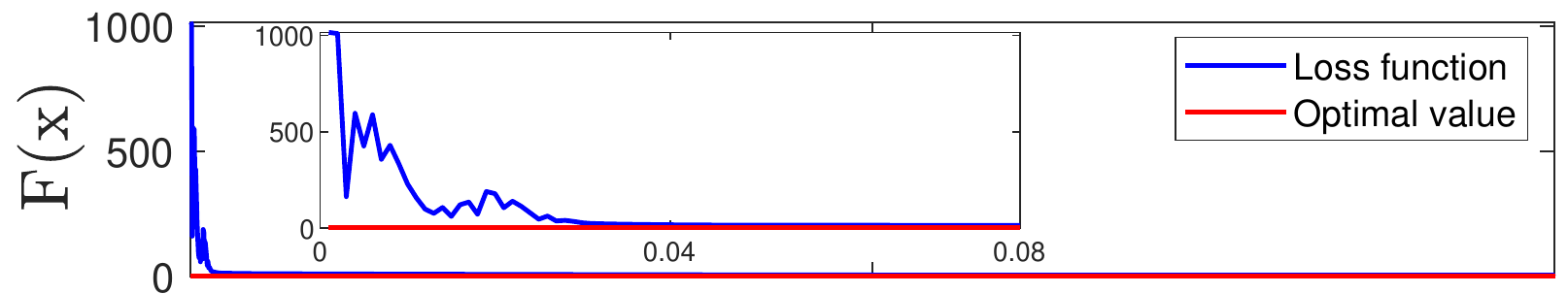}
	\includegraphics[width=3.5in]{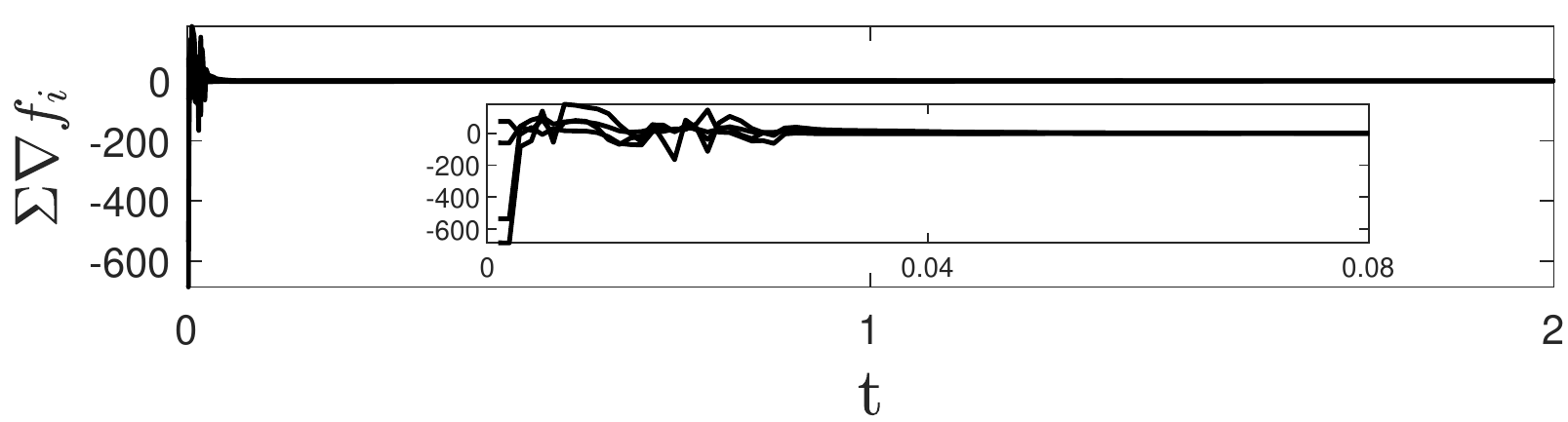}		
	\caption{The time-evolution of the SVM classifier parameters ~${\omega}_i$ and~$\nu_i$ (at all $5$ agents) under dynamics~\eqref{eq_xydot}-\eqref{eq_M}  along with overall loss function~$F(\mb{x})$ and sum of the gradients~$\sum_{i=1}^{5} \boldsymbol{ \nabla} f_i(\mb{x}_i)$. The optimal values based on the centralized SVM are also shown for comparison.}
	\label{fig_dynamics}	
\end{figure}
The agents reach consensus on the optimal value ${\overline{\mb{x}}^*=[\overline{\omega}(1),\overline{\omega}(2),\overline{\omega}(3),\overline{\nu}]^\top}$ as the parameters of the separating hyperplane in~$\mathbb{R}^3$. Via the inverse mapping, the hyperplane in~$\mathbb{R}^3$ represents an ellipse formulated as ${\overline{\omega}(1)z_1^2 + \overline{\omega}(2)z_2^2 -\overline{\nu}=0}$ ($z_1$ and~$z_2$ as the Cartesian coordinates in~$\mathbb{R}^2$), which separates the original data points~$\boldsymbol{\chi}_i$'s in~$\mathbb{R}^2$. The calculated separating ellipses by  all~$5$ agents are shown in Fig.~\ref{fig_2d_k} at two different time-instants.  
 \begin{figure}[]
 	\centering
 	\includegraphics[width=1.5in]{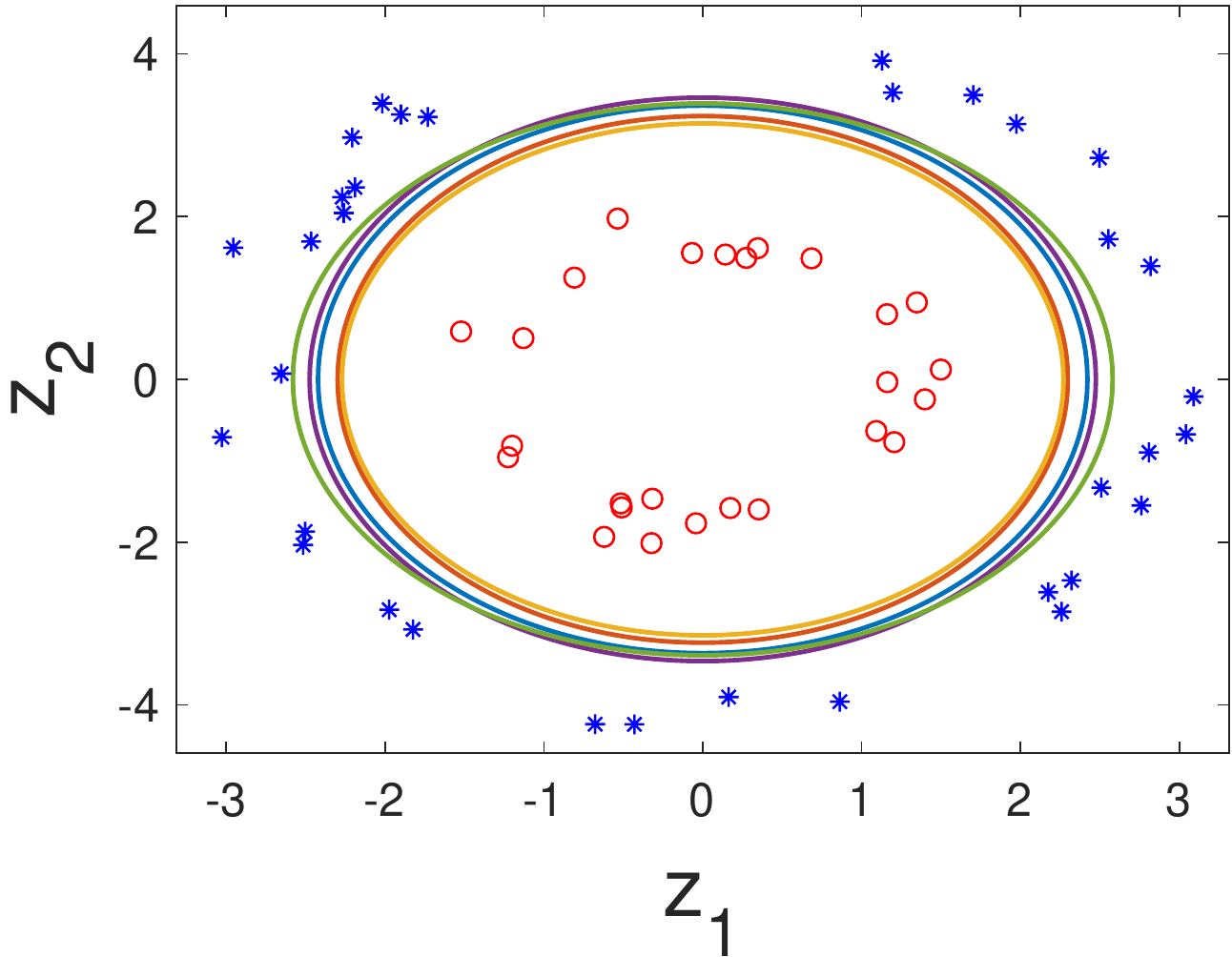}
 	\hspace{0.5cm}
 	\includegraphics[width=1.5in]{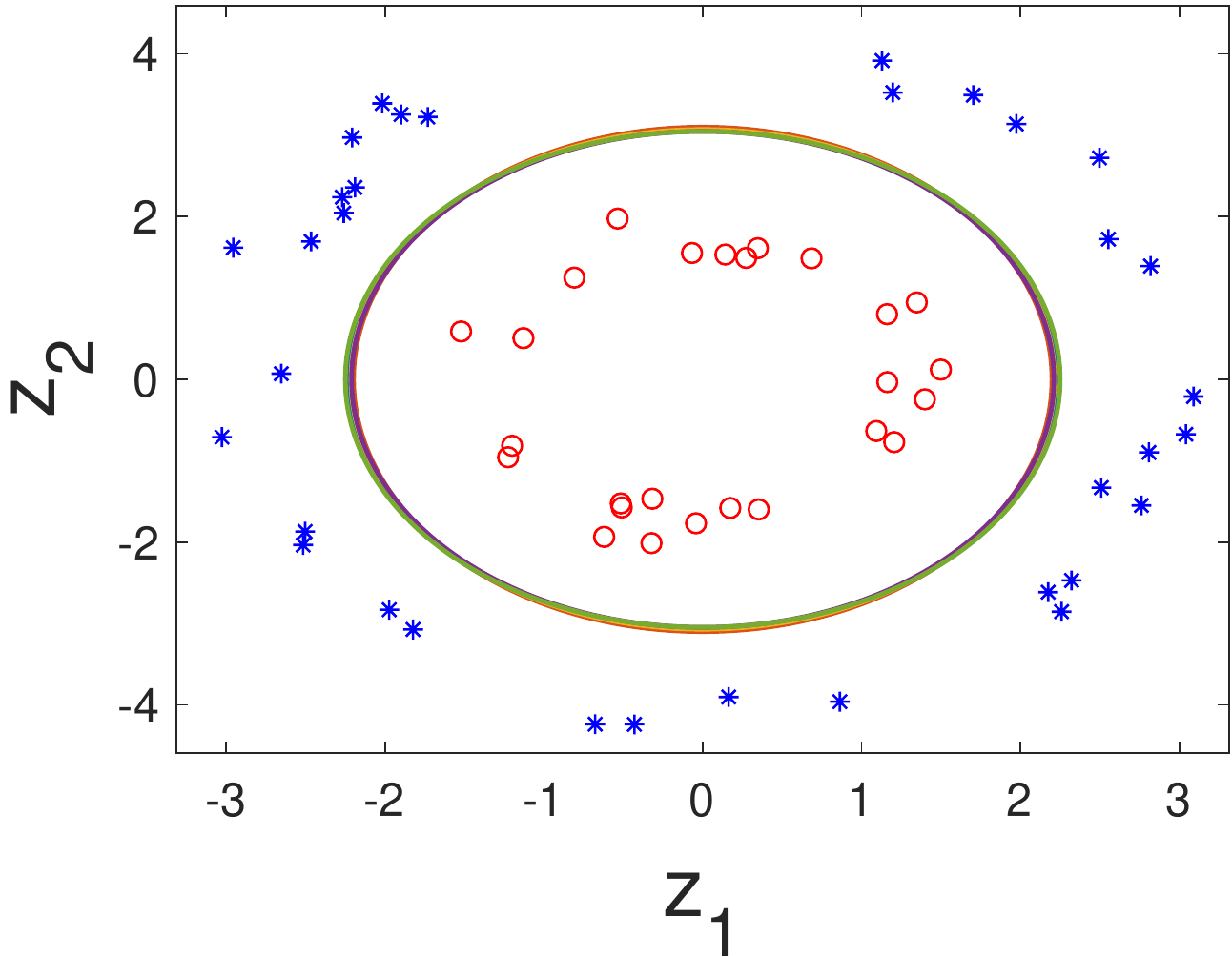}	
 	\caption{Optimal elliptical SVM classifiers computed by all the agents under the distributed optimization  dynamics~\eqref{eq_xydot}-\eqref{eq_M}  at time~$t=0.1$ (Left) and~${t=2}$ (Right).}
 	\label{fig_2d_k}
 \end{figure}
\section{Conclusion and Future Research} \label{sec_conclusion}
In this work, a Lipschitz dynamics is proposed to solve D-SVM over a dynamic WB-digraph in a hybrid setting. We adopt matrix perturbation analysis to prove convergence of the CT dynamics~\eqref{eq_xydot}-\eqref{eq_M} whose parameters vary due to  switching network topology in DT domain. In particular, our proposed distributed optimization in D-SVM setup enables the agents to cooperatively learn the classifier over a dynamic network via local information, improving classical D-SVM methods in terms of data privacy~\cite{navia2006distributed,chang2011psvm,bordes2005fast} and computational complexity~\cite{forero2010consensus}.

As future research direction, one can extend the results to the DT counterpart. 
%For example, using well-known explicit (Euler-Forward) or implicit (Euler-Backward and Tustin) methods, one can find DT approximation of the proposed CT dynamics~\eqref{eq_xydot}-\eqref{eq_M}. 
For example, for Euler-Forward method, the DT version of matrix $M$ in \eqref{eq_M} is $M_d = I+TM$ with~$T$ as the sampling period. Then, the explicit upper bound on~$T$ such that a stable CT dynamics from Theorem~\ref{thm_zeroeig}-\ref{thm_lyapunov} remains stable after discretization can be defined. Additionally, extensions to \textit{time-delayed} networks, \textit{online} D-SVM, and \textit{sparse} digraphs are directions of interest.

%Using well-known explicit (Euler-Forward) or implicit (Euler-Backward and Tustin) methods, one can find DT approximation of the proposed CT dynamics~\eqref{eq_xydot}-\eqref{eq_M}. For example, for Euler-Forward method, the DT version of matrix $M$ in \eqref{eq_M} is $M_d = I+TM$ with~$T$ as the sampling period. Then, the explicit upper bound on~$T$ such that an stable CT dynamics from Theorem~\ref{thm_zeroeig}-\ref{thm_lyapunov} remains stable after discretization can be defined from the results in~\cite{axelsson2014discrete}.
%As future research direction, one can extend the results to \textit{time-delayed} networks~\cite{hadjicostis2013average}, \textit{online} D-SVM~\cite{giannakis2020time}, and \textit{sparse} digraphs~\cite{TCNS_resource}.
%As future research direction, one can extend the results to  time-delayed networks based on~\cite{hadjicostis2013average}. Further, by leveraging the hybrid nature of our proposed solution, it can be extended to online D-SVM~\cite{giannakis2020time} where the training sets vary in time resulting in a dynamic loss function. 
%Currently, we are working to relax our network-connectivity requirement to sparsely-connected digraphs as in~\cite{TCNS_resource}.
%can be adopted for distributed optimization  such that the union of the directed networks over some bounded non-overlapping time-intervals is weight-balanced.

\bibliographystyle{IEEEbib}
\bibliography{bibliography}
\end{document}